\documentclass[a4paper,11pt,reqno]{amsart}
\usepackage[body={15.5cm,20.3cm}]{geometry}
\usepackage{amssymb}
\usepackage{hyperref}
\usepackage{enumerate}

\newcommand*{\mailto}[1]{\href{mailto:#1}{\nolinkurl{#1}}}
\newcommand{\hsp}[1]{\hspace*{#1mm}}

\numberwithin{equation}{section}
\swapnumbers
\theoremstyle{plain}
\newtheorem{thm}{Theorem}[section]
\newtheorem{cor}[thm]{Corollary}
\newtheorem{lem}[thm]{Lemma}
\theoremstyle{remark}
\newtheorem*{rem*}{Remark}
\newtheorem*{rems*}{Remarks}
\theoremstyle{definition}

\newtheorem*{assus*}{Assumptions}
\newtheorem*{Def*}{Definition}
\newtheorem*{not*}{Notation}

\providecommand{\C}{\mathcal}
\providecommand{\D}{\mathbb}

\newcommand{\full}{\mathrm{F}}
\newcommand{\loc}{\mathrm{loc}}

\newcommand{\unif}{\mathrm{unif}}
\newcommand{\weg}{\mathrm{W}}
\providecommand{\abs}[1]{\lvert#1\rvert}
\providecommand{\accol}[1]{\lbrace#1\rbrace}
\providecommand{\croch}[1]{\lbrack#1\rbrack}
\providecommand{\norm}[1]{\lVert#1\rVert}
\providecommand{\scal}[1]{\langle#1\rangle}
\DeclareMathOperator{\comp}{\textup{\textsf{c}}}
\DeclareMathOperator{\dist}{dist}
\DeclareMathOperator*{\essup}{ess\,sup}
\DeclareMathOperator{\expect}{\D{E}}

\DeclareMathOperator{\prob}{\D{P}}
\DeclareMathOperator{\ran}{Ran}
\DeclareMathOperator{\spec}{\sigma}

\DeclareMathOperator{\supp}{supp}
\DeclareMathOperator{\tr}{tr}

\newcommand{\Hmm}[1]{\leavevmode{\marginpar{\tiny%
$\hbox to 0mm{\hspace*{-0.5mm}$\leftarrow$\hss}%
\vcenter{\vrule depth 0.1mm height 0.1mm width \the\marginparwidth}%
\hbox to 0mm{\hss$\rightarrow$\hspace*{-0.5mm}}$\\\relax\raggedright #1}}}

\begin{document}

\title[Localization near fluctuation boundaries]{An uncertainty
principle, Wegner estimates and localization near fluctuation
boundaries}
\author[A. Boutet de Monvel]{Anne Boutet de Monvel$^{\star}$}
\address{$^{\star}$Institut de Math\'ematiques de Jussieu,
Universit\'e Paris Diderot Paris 7\\
175 rue du Chevaleret\\ 75013 Paris \\ France}
\email{\mailto{$^{\star}$aboutet@math.jussieu.fr}}
\urladdr{\url{http://www.math.jussieu.fr/~aboutet/}}
\author[D. Lenz]{Daniel Lenz$^{\dagger}$}
\address{$^{\dagger}$Mathematisches Institut,
Friedrich-Schiller-Universit\"at Jena, 07737 Jena, Germany}
\email{\mailto{daniel.lenz@uni-jena.de}}
\author[P. Stollmann]{Peter Stollmann$^{\ddagger}$}
\address{$^{\ddagger}$Fakult\"at f\"ur
   Mathematik, Technische Universit\"at, 09107 Chemnitz, Germany}
\email{\mailto{peter.stollmann@mathematik.tu-chemnitz.de}}
\date{21.4.2009}
\begin{abstract}
We prove a simple uncertainty principle and show that it can be
applied to prove Wegner estimates near fluctuation boundaries. This
gives new classes of models for which localization at low energies
can be proven.
\end{abstract}
\maketitle
\section*{Introduction}

Starting point of the present paper was the lamentable fact that for
certain random models with possibly quite small and irregular support
there was a proof of localization via fractional moment techniques
(at least for $d\leq 3$) but no proof of Wegner estimates necessary
for multiscale analysis. The classes of models include models with
surface type random potentials as well as Anderson models with
displacement (see \cite{BNSS-06}) but actually much more classes of
examples could be seen in the framework established there which was
labelled ``fluctuation boundaries''. Actually, the big issue in the
treatment of random perturbations with small or irregular support is
the question, whether the spectrum at low energies really feels the
random perturbation. This is exactly what is formalized in the
fluctuation boundary framework.

In the present paper we stablish the necessary Wegner estimates by
using the method from Combes, Hislop, and Klopp \cite{CHK-07} so that
we get the correct volume factor and the modulus of continuity of the
random variables. One of the main ideas we borrow from the last
mentioned work is to show that spectral projectors are ``spread
out'', a property we call ``uncertainty principle''.

The solution to the above mentioned problem is now quite simple in
fact. In an abstract framework we show that such an uncertainty
principle of the form
\begin{equation}          \label{eq.intro.uncertainty.principle}
P_I(H_0)WP_I(H_0)\geq\kappa P_I(H_0),
\end{equation}
where $W\geq 0$ is bounded and $P_I(H_0)$ denotes the spectral
projection, is in a sense equivalent to the mobility of
\begin{equation}          \label{eq.inf.perturb.spectrum}
\lambda(t):=\inf\spec(H_0+tW).
\end{equation}
This is done in Section 1.

That fits perfectly with the fluctuation boundary concept and gives
the appropriate Wegner estimates. Actually, if the integrated density
of states exists, it then must be continuous, provided the
distribution of the random variables has a common modulus of
continuity. We will prove this in Section 2.

Finally, in Section 3 we show how to exploit these Wegner estimates
for a proof of localization. It lies in the nature of these different
methods that we thus get localization under less restrictive
conditions than what was needed in \cite{BNSS-06}. One main point is
the dimension restriction of the latter paper, $d\leq 3$, which
certainly is not essential but is essential for a proof of digestable
length. Clearly, the estimates one gets via the fractional moment
method are more powerful.

\subsection*{Acknowledgement}
Fruitful discussions with G.~Stolz are gratefully acknowledged. Part
of this work was done at a stay of D.L. and P.S. at Paris, financial
support by the DFG (German Science Foundation) and the University
Paris Diderot Paris 7 are gratefully acknowledged.

\section{An uncertainty principle and mobility of the ground state
energy}                   \label{sec:uncertainty.principle}

In this section we fix a rather abstract setting:
$\C{H}$ is a Hilbert space,
$H_0$ is a selfadjoint operator in $\C{H}$ with
\begin{equation}
\lambda(0):=\inf\spec(H_0)>-\infty.
\end{equation}
Moreover, $W$ is assumed to be bounded and nonnegative.

The uncertainty principle we want to study is the existence of a
positive $\kappa$ such that
\begin{equation}      \label{eq:uncertainty}
P_IWP_I\geq\kappa P_I\tag{$\star$}
\end{equation}
where $I\subset\D{R}$ is some compact interval, $I=\croch{\min I,\max
I}$ and $P_I=P_I(H_0)=\chi_I(H_0)$ is the corresponding spectral
projection.

It is reasonable to call \eqref{eq:uncertainty} an uncertainty
principle as a state in the range of $P_I$ cannot be ``concentrated
where $W$ vanishes''.

In our main application, $H_0$ will be a Schr\"odinger operator so
that \eqref{eq:uncertainty} is in fact a variant of the usual
uncertainty principle, at least for $H_0=-\Delta$.

The use of \eqref{eq:uncertainty} for the proof of Wegner estimates
is due to Combes, Hislop, and Klopp, see \cite{CHK-03,CHK-07}. Its
importance lies in the fact that it takes care of random potentials
with small support. Our purpose here is to prove a simple criterion
that implies \eqref{eq:uncertainty} and can be checked rather easily.

\begin{thm}                \label{thm.uncertainty.principle}
Let for $t\geq 0$ 
\begin{equation}
\lambda(t):=\inf_{}\spec(H_0+tW)
\end{equation}
and assume that $\lambda(t_0)>\max I$ for some $t_0>0$. Then
\begin{equation}
P_IWP_I\geq\biggl\lbrack\sup_{t>0}\frac{\lambda(t)-\max I}{t}\biggr\rbrack P_I.
\end{equation}
\end{thm}

Of course, the assumption in the theorem is merely there to guarantee
that the square bracket is positive!

\begin{proof}
Assume that \eqref{eq:uncertainty} does not hold for some $\kappa>0$.
Then we find $g\in\ran P_I$ with $\norm{g}=1$ and
\[
\scal{Wg,g}=\scal{P_IWP_Ig,g}<\kappa.
\]
Since $\scal{H_0g,g}\leq\max I$, by the functional calculus, we get,
for any $t>0$,
\[
\lambda(t)\leq\scal{(H_0+tW)g,g}<\max I+t\,\kappa,
\]
which implies
\[
\kappa>\frac{\lambda(t)-\max I}{t}\, .
\]
By contraposition, we get the assertion.
\end{proof}

\begin{rems*}
\begin{enumerate}[(1)]
\item
One particularly nice aspect of the above result is that the
important constant is controlled in a simple way.
\item
Once the ground state energy is pushed up by $W$ we get an
uncertainty principle \eqref{eq:uncertainty} at least for intervals
$I$ near $\lambda(0)$.
\item
The corresponding uncertainty result in \cite{CHK-03} for periodic
Schr\"odinger operators does not follow from the preceding theorem.
\end{enumerate}
\end{rems*}

There is a kind of converse to Theorem \ref{thm.uncertainty.principle}.

\begin{lem}                \label{lem.uncertainty.principle}
If \eqref{eq:uncertainty} holds for $I$ with $\min I=\lambda(0)=\inf \sigma (H_0)$ and
$\max I>\min I$, then
\[
\lambda(t)>\lambda(0)
\]
for all $t>0$.
\end{lem}

\begin{proof}
We only need to consider small $t>0$ since $W\geq 0$. For $f\in
D(H_0)$, $\norm{f}=1$, let $f_1:=P_If$ and $f_2:=P_{I^{\comp}}f$ so
that $\norm{f_1}^2+\norm{f_2}^2=1$. We consider
\begin{align*}
\scal{(H_0+tW)f,f}
&=\scal{H_0f_1,f_1}+\scal{H_0f_2,f_2}+t\scal{Wf,f}\\
&\geq(\max I)\norm{f_2}^2+\lambda(0)\norm{f_1}^2+t\kappa\,\norm{f_1}^2-2t\,\norm{W}\,\norm{f_1}\,\norm{f_2}\\
&\geq\lambda(0)\norm{f}^2+(\max
I-\lambda(0))\norm{f_2}^2-2t\,\norm{f_1}\,\norm{f_2}\,\norm{W}+t\kappa\,\norm{f_1}^2.
\end{align*}
A knowing smile at the last quadratic (!) expression in $t$ reveals
that it is strictly larger than $\lambda(0)$ for $t$ small enough.
\end{proof}

\section{Continuity of the IDS near weak fluctuation boundaries}
\label{sec:ids.weak.fluctu}

The main result here is, in fact, rather an ``optimal'' Wegner
estimate meaning that we recover at least the modulus of continuity
of the random variables in the Wegner estimate as well as the correct
volume factor. The models we consider needn't have a homogeneous
background so that the integrated density of states, IDS need not
exist. See \cite{Ves-08} for a recent survey on how to prove the existence of the IDS in various different settings.  We show that a straightforward application of Theorem
\ref{thm.uncertainty.principle} above gives the necessary input to
perform the analysis of \cite{CHK-07} in a rather general setting
which we are going to introduce now.
\begin{enumerate}[{(A}1)]
\item
The background potential $V_0\in L_{\loc,\unif}^p(\D{R}^d)$ with
$p=2$ if $d\leq 3$, and $p>\frac{d}{2}$ if $d>3$.
\item
The set $\C{I}\subset\D{R}^d$, where the random impurities are
located, is uniformly discrete, i.e.,
\[
\inf_{\substack{\alpha,\beta\in\C{I}\\\alpha\neq\beta}}\abs{\alpha-\beta}=:r_{\C{I}}>0.
\]
\end{enumerate}
\begin{enumerate}[\hsp{8}]
\item[$(\widetilde{\mathrm{A}3})$]
For the probability measure $\prob$ on
$\Omega=\prod_{\alpha\in\C{I}}\croch{0,\eta_{\max}}$ we use conditional probabilities to define the 
following uniform bound
\[
s(\varepsilon)=\sup_{\alpha}\essup_{E\in\D{R}}\essup_{(\omega_{\beta})_{\beta\neq\alpha}}\prob\bigl\lbrace\omega_\alpha\in\croch{E,E+\varepsilon}\mid(\omega_{\beta})_{\beta\neq\alpha}\bigr\rbrace.
\]
\item[$(\widetilde{\mathrm{A}4})$]
Let $E_0:=\inf\spec(H_0)$ and let
\[
H_{\full}:=H_0+\eta_{\max}\sum_{\alpha\in\C{I}}U_{\alpha}
\]
the subscript $\full$ standing for ``full coupling''.
The single site potentials $U_{\alpha}$, $\alpha\in\C{I}$ are
measurable functions on $\D{R}^d$ that satisfy
\[
c_U\chi_{\Lambda_{r_U}(\alpha)}\leq U_{\alpha}\leq
C_U\chi_{\Lambda_{R_U}(\alpha)}
\]
for all $\alpha\in\C{I}$, with $c_U$, $C_U$, $r_U$, $R_U>0$
independent of $\alpha$. Here, $\Lambda_s (\beta)$ denotes the box with sidelength $2 s$ and center $\beta$.  
\[
V_{\omega}(x)=\sum_{\alpha\in\C{I}}\omega_{\alpha}U_{\alpha}(x)
\]
and
\[
H:=H(\omega):=H_0+V_{\omega}\text{ in }L^2(\D{R}^d).
\]
Assume that $E_0$ is a \emph{weak fluctuation boundary} in the sense
that $E_{\full}:=\inf\spec(H_{\full})>E_0$.
\end{enumerate}

\begin{rems*}
\begin{enumerate}[(1)]
\item
In \cite{BNSS-06} (A3) and (A4) are stronger than their counterparts
$(\widetilde{\mathrm{A}3})$ (which actually isn't an assumption at
all) and $(\widetilde{\mathrm{A}4})$ here.
\item
The modulus of continuity $s(\,\cdot\,)$ from (A3) also appears in
\cite{CHK-07}, where, however, the variables appearing in the
conditional probabilities are not displayed correctly. For a detailed
discussion of regular conditional probabilities see, e.g.,
\cite{Kln-08}.
\end{enumerate}
\end{rems*}

Wegner estimates, named after Wegner's original work \cite{Weg-81}, are an important tool in
random operator theory. They give a bound on the probability that the
eigenvalues of a local Hamiltonian come close to a given energy.
For a list of some recent papers, see \cite{Chulaevsky-07, Chulaevsky-08, ChulaevskyS-08,CHK-03, CHK-07, HundertmarkKNSV-06, Krish-07} and the account in the recent Lecture Notes Volume \cite{Ves-08}.
We
consider a box $\Lambda\subset\D{R}^d$ and denote by
$H_{\Lambda}(\omega)$ the restriction of $H(\omega)$ to
$L^2(\Lambda)$ with Dirichlet boundary conditions and with
$H_0^{\Lambda}$ the restriction of $H_0$ to $L^2(\Lambda)$ with
Dirichlet boundary conditions. Here comes our main application of
Theorem \ref{thm.uncertainty.principle}:

\begin{thm}                \label{thm.wegner}
Assume \emph{(A1)-(A2)} and $(\widetilde{\mathrm{A}3})$-$(\widetilde{\mathrm{A}4})$.
Then, for every $\delta>0$ there exists a constant
$C_{\weg}=C_{\weg}(\delta)$ such that for any interval
$I=\croch{E_0,E_{\full}-\delta}$ we have:
\begin{align*}
\prob\bigl\lbrace\spec\bigl(H_{\Lambda}(\omega)\bigr)\cap
I\neq\varnothing\bigr\rbrace
&\leq\expect\bigl\lbrace\tr P_I\bigl(H_{\Lambda}(\omega)\bigr)\bigr\rbrace\\
&\leq C_{\weg}\cdot\abs{\Lambda}\cdot s(\varepsilon).
\end{align*}
\end{thm}

Clearly, any application will need some further assumptions on
$s(\,\cdot\,)$ for which we a priori just know that $0\leq
s(\varepsilon)\leq 1$ for all $\varepsilon>0$.

\begin{proof}
We rely on the analysis from \cite{CHK-07}. The main point is to find
an estimate
\begin{equation}   \label{eq:stars}
P_I(H_0^{\Lambda})W_{\Lambda}P_I(H_0^{\Lambda})\geq\kappa
P_I(H_0^{\Lambda})\tag{$\star\star$}
\end{equation}
with a constant $\kappa$ independent of $\Lambda$ and $I$ (as long as
$I\subset\croch{E_0,E_{\full}-\delta}$), and 
\[
W_{\Lambda}:=\Bigl(\sum_{\alpha\in\C{I}}U_{\alpha}\Bigr)\cdot\chi_{\Lambda}.
\]
Once \eqref{eq:stars} is established, the proof of \cite[Theorem
1.3]{CHK-07} takes over, with minor modifications of notation.

But \eqref{eq:stars} follows easily from Theorem
\ref{thm.uncertainty.principle} and $(\widetilde{\mathrm{A}4})$: As
Dirichlet boundary conditions shift the spectrum up, for any
$t\geq\eta_{\max}$:
\[
\lambda(t)=\inf\spec(H_0^{\Lambda}+tW_{\Lambda})\geq\inf\spec(H_0+tW)\geq
E_{\full}.
\]
For $I\subset\croch{E_0,E_{\full}-\delta}$ we see that
\[
\lambda(\eta_{\max})-\max I>\delta
\]
so that we get an uncertainty inequality with
\[
\kappa=\frac{\delta}{\eta_{\max}}\,.\qedhere
\]
\end{proof}

\begin{rem*}
We should point out that the input from \cite{CHK-07} is rather
substantial. While the uncertainty principle is important to deal
with possibly small support, there is also the issue of the correct
volume factor which is settled in \cite{CHK-07}.
\end{rem*}

Like in the latter paper, if we assume on top that the IDS
$N(\,\cdot\,)$ of the random operator $H$ exists, then the preceding
theorem implies that $N(\,\cdot\,)$ is as continuous as
$s(\,\cdot\,)$ is.

\begin{cor}                \label{cor.uncertainty.principle}
Assume \emph{(A1)-(A2)} and
\emph{$(\widetilde{\mathrm{A}3})$-$(\widetilde{\mathrm{A}4})$}, and,
additionally that the IDS $N(\,\cdot\,)$ of $H$ exists. Then there
exists a locally bounded $c_{\weg}(\,\cdot\,)$ on $\lbrack
E_0,E_{\full})$ such that
\[
N(E+\varepsilon)-N(E)\leq c_{\weg}(E)\cdot s(\varepsilon)
\]
for $\varepsilon$ small enough. In particular, $N(\,\cdot\,)$ is
continuous on $\lbrack E_0,E_{\full})$, whenever $s(\varepsilon)\to
0$ as $\varepsilon\to 0$.
\end{cor}

\section{Localization near fluctuation boundaries}   \label{sec:localization}

As mentioned already in the introduction, the validity of a Wegner
estimate was missing for a proof of localization via multiscale
analysis. Due to Theorem \ref{thm.wegner}, this is now resolved. The
assumptions we need to make now are weaker than what is found in
\cite{BNSS-06} but stronger than what we needed in the preceding
section.
\begin{enumerate}[\hsp{8}]
\item[$(\overline{\mathrm{A}3})$]
The random variables $\eta_{\alpha}\colon\Omega\to\D{R}$,
$\omega\mapsto\omega_{\alpha}$ are independent, supported in
$\croch{0,\eta_{\max}}$ and the modulus of continuity
\[
s(\varepsilon):=\sup_{\alpha\in
I}\sup_{E\in\D{R}}\prob\accol{\eta_{\alpha}\in\croch{E,E+\varepsilon}}
\]satisfies
\[
s(\varepsilon)\leq(-\ln\varepsilon)^{-\alpha}
\]
for some $\alpha>\frac{4d}{2-m}$, where $m\in(0,2)$ is as in
$(\overline{\mathrm{A}4})$.
\item[$(\overline{\mathrm{A}4})$]
Additionally to $(\widetilde{\mathrm{A}4})$ assume that there exists
$m\in(0,2)$ and $L^*$ such that for some $\xi>0$, all $L\geq L^*$ and
all $x\in\D{Z}^d$:
\[
\prob\bigl\lbrace\spec\bigl(H_{\Lambda_L(x)}(\omega)\bigr)\cap\croch{E_0,E_0+L^{-m}}\neq\varnothing\bigr\rbrace\leq
L^{-\xi}.
\]
\end{enumerate}

\begin{rem*}
Cleary, the assumptions (A1)-(A4) from \cite{BNSS-06} imply (A1)-(A2)
and $(\overline{\mathrm{A}3})$-$(\overline{\mathrm{A}4})$ so that the
localization result below extends the localization result from the
latter paper.
\end{rem*}

\begin{thm}                \label{thm.localization}
Assume \emph{(A1)-(A2)} and
\emph{$(\overline{\mathrm{A}3})$-$(\overline{\mathrm{A}4})$}. Then
there is a $\delta>0$ such that in $\croch{E_0,E_0+\delta}$ the
spectrum of $H(\omega)$ is pure point $\prob$-a.s. Moreover, for $p$
small enough and $\eta\in L^{\infty}$ with
$\supp\eta\subset\croch{E_0,E_0+\delta}$ it follows that
\[
\expect\bigl\lbrace\bigl\lVert\abs{X}^p\eta\bigl(H(\omega)\bigr)\cdot\chi_K\bigr\rVert\bigr\rbrace<\infty
\]
for every compact $K\subset\D{R}^d$.
\end{thm}

\begin{rems*}
\begin{enumerate}[(1)]
\item
Maybe one can strengthen the estimate of Theorem
\ref{thm.localization} in the sense of \cite{GK-01}. Note, however,
that in the latter paper a stronger Wegner estimate is supposed to
hold.
\item  The theorem provides an extension to $d\geq 4$ of  the main result of \cite{BNSS-06}. 
Moreover,  there is no technique at the
moment to include single site distributions as singular as the ones
allowed here in the fractional moment methods. In these aspects, our
result considerably  extends the main result of \cite{BNSS-06}. 
\item
At the same time, the estimates that come out of our analysis are
weaker than those in the latter paper.
\end{enumerate}
\end{rems*}

\begin{proof}[Sketch of the proof]
We use the multiscale setup from \cite{Stoll-01}. By now it is quite
well understood that homogeneity doesn't play a major role so that
multiscale analysis goes through without much alterations if we can
verify the necessary input, i.e., Wegner estimates and initial length
scale estimates.

Let us begin with the latter: Combes\textendash Thomas estimates give
that $(\overline{\mathrm{A}4})$ implies an initial estimate of the
form $G(I,\ell,\gamma,\xi)$ with $\xi$ from
$(\overline{\mathrm{A}4})$,
$I_{\ell}=\croch{E_0,E_0+\frac{1}{2}\,\ell^{-m}}$,
$\gamma_{\ell}=\ell^{-\frac{m}{2}}$ so that the exponent is of the
form $\gamma_{\ell}=\ell^{\beta-1}$ with
$\beta=\beta_m=\frac{2-m}{2}$.

We have to check that an appropriate Wegner estimate is valid as
well, i.e., that, for some $q>d$, $\theta<\frac{\beta}{2}$ we have
that
\[
\prob\bigl\lbrace\dist\bigl(\spec\bigl(H_{\Lambda}(\omega)\bigr),E\bigr)\leq\exp(-L^{\theta})\bigr\rbrace\leq
L^{-q}
\]
for $L$ large enough. We check that
\begin{align*}
\prob\bigl\lbrace\dist\bigl(\spec\bigl(H_{\Lambda}(\omega)\bigr),E\bigr)\leq\exp(-L^{\theta})\bigr\rbrace
&\leq s\bigl(2\exp(-L^{\theta})\bigr)\\
&\leq\bigl\lbrack-\ln\bigl(2\exp(-L^{\theta})\bigr)\bigr\rbrack^{-\alpha}\\
&=(-\ln 2+L^{\theta})^{-\alpha}\\
&\sim L^{-\theta\alpha}
\leq L^{-\theta\frac{4d}{2-m}}\\
&=L^{-(\frac{2-m}{4}+\kappa)\frac{4d}{2-m}}
=L^{-d-\kappa\frac{4d}{2-m}}
\end{align*}
where we have chosen $\theta=\frac{\beta}{2}-x=\frac{2-m}{4}-\kappa$
with positive $\kappa$. Then, the Wegner estimate is fulfilled for
$q=d+\kappa\frac{4d}{2-m}$.

The appropriate $p$ in the strong dynamical localization estimate can
be chosen at most
$\inf\bigl\lbrace\kappa\frac{4d}{2-m},\xi\bigr\rbrace$ with $\xi$
from $(\overline{\mathrm{A}3})$.

An appeal to \cite[Theorems 3.2.2 and 3.4.1]{Stoll-01} gives the result.
\end{proof}


\end{document}